\documentclass[a4paper,UKenglish]{lipics-v2021}

\usepackage{xspace}
\usepackage[textsize=small]{todonotes}
\usepackage{stmaryrd}
\usepackage{centernot}
\usepackage{cleveref}
\usepackage{algorithm}
\usepackage[noend]{algpseudocode}

\usepackage{tikz} 
\usetikzlibrary{arrows, automata, positioning,calc}
\usetikzlibrary{shapes,snakes}

\newcommand{\ignore}[1]{}

\newcommand{\Oo}{\mathcal{O}}

\newcommand{\F}{\mathcal{F}}

\newcommand{\B}{\mathcal{B}}

\newcommand{\N}{\mathbb{N}}
\newcommand{\Z}{\mathbb{Z}}
\newcommand{\R}{\mathbb{R}}
\newcommand{\Q}{\mathbb{Q}}

\newcommand{\pre}{\textsc{pre}}
\newcommand{\post}{\textsc{post}}

\newcommand{\trans}[1]{\stackrel{#1}{\longrightarrow}}

\newcommand{\trg}{\textup{trg}}
\newcommand{\src}{\textup{src}}

\newcommand{\zero}{\textsc{zero}}

\newcommand{\norm}{\textup{norm}}
\newcommand{\supp}{\textup{supp}}
\newcommand{\trash}{\textup{trash}}
\newcommand{\proj}{\textup{proj}}

\newcommand{\reaches}{\longrightarrow}
\newcommand{\nreaches}{\centernot\longrightarrow}
\newcommand{\lin}{\textup{Lin}}

\newcommand{\inp}{\textup{in}}
\newcommand{\out}{\textup{out}}
\newcommand{\last}{\textup{last}}
\newcommand{\decr}{\textup{dec}}
\newcommand{\state}{\textup{state}}
\newcommand{\nemp}{{\neq\emptyset}}

\DeclareSymbolFont{extraup}{U}{zavm}{m}{n}
\DeclareMathSymbol{\varheart}{\mathalpha}{extraup}{86} 
\DeclareMathSymbol{\vardiamond}{\mathalpha}{extraup}{87} 

\newcommand{\kqhide}[1]{\textcolor{teal}{sth old is hidden here}}

\newcommand{\nl}{\textup{NL}\xspace}

\newcommand{\pspace}{\textup{PSpace}\xspace}

\newcommand{\expspace}{\textup{ExpSpace}\xspace}

\newcommand{\tower}{\textup{Tower}\xspace}
\newcommand{\ackermann}{\textup{Ackermann}\xspace}

\newcommand{\vr}[1]{#1}

\title{Long Runs Imply Big Separators in Vector Addition Systems}

\author{Wojciech Czerwi\'nski}{University of Warsaw}{wczerwin@mimuw.edu.pl}{0000-0002-6169-868X}{Supported by the ERC grant INFSYS, agreement no. 950398.}
\author{Adam J\k{e}drych}{University of Warsaw}{aj370982@students.mimuw.edu.pl}{}{}
\Copyright{Wojciech Czerwi\'nski, Adam J\k{e}drych}
\authorrunning{W. Czerwi\'nski and A. J\k{e}drych}


\relatedversion{}

\EventEditors{John Q. Open and Joan R. Access}
\EventNoEds{2}
\EventLongTitle{42nd Conference on Very Important Topics (CVIT 2016)}
\EventShortTitle{CVIT 2016}
\EventAcronym{CVIT}
\EventYear{2016}
\EventDate{December 24--27, 2016}
\EventLocation{Little Whinging, United Kingdom}
\EventLogo{}
\SeriesVolume{42}
\ArticleNo{23}

\begin{document}

\maketitle

\nolinenumbers

\begin{abstract}
Despite recent progress which settled the complexity of the reachability problem for Vector Addition Systems with States (VASSes)
as being Ackermann-complete we still lack much understanding for that problem.
A striking example is the reachability problem for three-dimensional
VASSes (3-VASSes): it is only known to be PSpace-hard and not known to be elementary.
One possible approach which turned out to be successful for many VASS subclasses is to prove that to check reachability it suffices to inspect only runs of some bounded length. This approach however has its limitations, it is usually hard to design an algorithm
substantially faster than the possible size of finite reachability sets in that VASS subclass. It motivates a search for other
techniques, which may be suitable for designing fast algorithms.
In 2010 Leroux has proven that non-reachability between two configurations implies separability of the source from the target
by some semilinear set, which is an inductive invariant. There can be a reasonable hope that it suffices to look for separators
of bounded size, which would deliver an efficient algorithm for VASS reachability. In the paper we show that
also this approach meets an obstacle: in VASSes fulfilling some rather natural conditions
existence of only long runs between some configurations implies
existence of only big separators between some other configurations (and in a slightly modified VASS). Additionally
we prove that a few known examples of involved VASSes fulfil the mentioned conditions. Therefore improving the complexity of the
reachability problem (for any subclass) using the separators approach may not be simpler than using the short run approach.
\end{abstract}

\section{Introduction}
The complexity of the reachability problem for Vector Addition Systems (VASes) was a challenging and natural problem
for a few decades of research. The first result about its complexity
was \expspace-hardness by Lipton in 1976~\cite{Lipton76}.
Decidability was shown by Mayr in 1981~\cite{DBLP:conf/stoc/Mayr81}.
Later the construction of Mayr was further simplified and presented in a bit different light
by Kosaraju~\cite{DBLP:conf/stoc/Kosaraju82} and Lambert~\cite{DBLP:journals/tcs/Lambert92}.
This construction is currently known under the name KLM decomposition after the three main inventors.
The first complexity upper bound was obtained by Leroux and Schmitz as recent as in 2016~\cite{DBLP:conf/lics/LerouxS15},
where the reachability problem was shown to be solvable in cubic-Ackermann time. A few years later
the same authors have shown that the problem can be solved in Ackermann time and actually in primitive
recursive time, when the dimension is fixed~\cite{DBLP:conf/lics/LerouxS19}.
At the same time a lower bound of \tower-hardness was established for the reachability
problem~\cite{DBLP:conf/stoc/CzerwinskiLLLM19,DBLP:journals/jacm/CzerwinskiLLLM21}. Last year two independent
papers have shown \ackermann-hardness of the reachability
problem~\cite{DBLP:conf/focs/Leroux21, DBLP:conf/focs/CzerwinskiO21}
thus settling the complexity of the problem to be \ackermann-complete.

However, despite this recent huge progress we still lack much understanding about the
reachability problem in VASSes. The most striking example is the problem for dimension three.
The best known lower bound for the reachability problem for 3-VASSes
is \pspace-hardness inherited from the hardness result for 2-VASSes~\cite{DBLP:conf/lics/BlondinFGHM15},
while the best known upper bound is much bigger then \tower. Concretely speaking the problem
can be solved in time $\F_7$~\cite{DBLP:conf/lics/LerouxS19}, where $\F_\alpha$ is the hierarchy
of fast-growing complexity classes, see~\cite{DBLP:journals/toct/Schmitz16} (recall that $\tower = \F_3$).
Similarly for other low dimensional VASSes the situation is unclear.
Currently the smallest dimension $d$ for which an \expspace-hardness
is known for $d$-VASSes is $d = 6$ and for which a \tower-hardness is known is $d = 8$~\cite{DBLP:journals/corr/abs-2203-04243}.
Therefore for any $d \in [3,5]$ the problem can be in \pspace and is not known to be elementary
and for $d \in [6,7]$ the problem can be elementary.
For some of those dimensions one can hope to get better hardness results in the future,
but we conjecture that for $3$-VASSes and $4$-VASSes the reachability problem actually is elementary
and the challenge is to find this algorithm.
This complexity gap is just one witness of our lack of understanding of the structure of VASSes.
One can hope that in the future we will be able to design efficient algorithms for the reachability problem
even for VASSes in high dimensions under the condition that they belong to some subclass, for example
they avoid some hard patterns. This can be quite important from a practical point of view. Understanding
for which classes efficient (say \pspace) algorithms are possible may be thus another future goal.
Therefore we think that the quest for better understanding the reachability problem is still valid,
even though the complexity of the reachability problem for general VASSes is settled.

A common and very often successful approach to the reachability problem is proving a short run property,
namely showing that in order to decide reachability it suffices to inspect only runs of some bounded length.
This technique was exploited a lot in the area of VASSes.
Rackoff proved his \expspace upper bound on the complexity of the coverability problem~\cite{DBLP:journals/tcs/Rackoff78}
using this approach: he has shown that if there exists a covering run, then there exists also a covering run
of at most doubly-exponential length.
Recently there was a lot of research about low dimensional VASSes and the technique of bounded run length was also
used there. In~\cite{DBLP:conf/lics/BlondinFGHM15} authors established complexity of the reachability problem
for 2-VASSes to \pspace-completeness by proving that in order to decide reachability it suffices to inspect runs
of at most exponential length. Similarly in the next paper in this line of research~\cite{DBLP:conf/lics/EnglertLT16}
it was shown that in 2-VASSes with unary transition representation it suffices to consider runs of polynomial length,
which proves \nl-completeness of the reachability problem.

This approach however meets a subtle obstacle when one tries to prove some upper bound on the complexity
of the reachability problem in say $d$-VASSes.
In order to show a bounded length property it is natural to try to unpump long runs in some way. Unpumping however
can be very tricky when the run is close to some of the axes, as only a small modification of the run may cause some counters to
become negative. A common approach to that problem is to modify a run when all its counters are high and therefore local run
modifications cannot cause any problem. Such an approach is used for example in 2-VASSes~\cite{DBLP:conf/lics/BlondinFGHM15}
and in the KLM decomposition~\cite{DBLP:conf/stoc/Kosaraju82, DBLP:journals/tcs/Lambert92, DBLP:conf/stoc/Mayr81}.
However, not all the runs may have a configuration with all counter values high. Therefore it is very convenient to have
a cycle, which increases all the counters simultaneously. Observe that existence of such a cycle implies that
the set of reachable configurations is infinite. Thus using this approach is challenging in the case when the reachability set is finite.
For finite set in turn it is hard to design algorithm substantially faster than the size of the reachability set.
It is well known that in $d$-VASSes finite reachability sets can be as large as $F_{d-1}(n)$, where
$n$ is the VASS size and $F_\alpha$ is the hierarchy of very fast growing functions, see~\cite{DBLP:journals/toct/Schmitz16}.
Therefore designing an algorithm breaking $F_{d-1}(n)$ time for the reachability problem in $d$-VASSes
(if such exist) may need some other approach. Breaking the barrier of the size of finite reachability set is possible in general,
but probably in many cases very challenging. To the best of our knowledge the only nontrivial algorithm breaking it 
is the one in~\cite{DBLP:conf/lics/EnglertLT16} doing a sophisticated analysis of possible behaviours of runs in 2-VASSes.
This motivates a search for other techniques, which may be more suitable for designing fast algorithms.

Leroux in his work~\cite{DBLP:conf/lics/Leroux09} provided an algorithm for the reachability problem,
which follows a completely different direction. He has shown that if there is no run from a configuration $s \in \N^d$
to a configuration $t \in \N^d$ in a VAS $V$ then there exists a semilinear set $S$, called a \emph{separator}, which
(1) contains $s$, (2) does not contain $t$ and (3) is an \emph{inductive invariant}, namely if $v \in S$ then also $v+t \in S$
for any transition $t$ of $V$ such that $v+t \in \N^d$. Notice that existence of a separator clearly implies non-reachability between
$s$ and $t$, so by Leroux's work non-reachability and existence of separator are equivalent. Then the following simple algorithm
decides whether $t$ is reachable from $s$: run two semi-procedures, one looks for possible runs between $s$ and $t$, longer and longer,
another one looks for possible separators between $s$ and $t$, bigger and bigger. Clearly either there exists a run or there exists
a separator, so at some point the algorithm will find it and terminate. From this perspective one can view runs and separators as dual
objects. If bounding the length of a run is nontrivial, maybe bounding the size of a separator can be another promising approach.
Notice that proving for example an $F_d(n)$ upper bound on the size of separators would provide an algorithm solving
the reachability problem which works in time at most exponential $F_d(n)$, which is still $F_d(\Oo(n))$ for $d \geq 3$.

\subparagraph*{Our contribution}
Our main contribution are two theorems stating that in VASSes fulfilling certain rather natural conditions if there are only long runs between
some of its two configurations then in a small modification of these VASSes there are only big separators for some other two configurations.
We designed Theorem~\ref{thm:simple} to have relatively simple statement, but also to be sufficiently strong for our applications.
Theorem~\ref{thm:advanced} needs more sophisticated notions and more advanced tools to be proven, but it has potentially a broader
spectrum of applications. 

Additionally we have shown that two nontrivial constructions of involved VASSes,
namely the 4-VASS from~\cite{DBLP:conf/concur/Czerwinski0LLM20} and VASS used in the \tower-hardness
construction from~\cite{DBLP:conf/stoc/CzerwinskiLLLM19} fulfil the conditions proposed by us in Theorem~\ref{thm:simple}.
This indicates that for each VASS subclass $\F$ (for example $3$-VASSes) either (1) in order to prove better upper complexity bound
for the reachability problem in $\F$ one should focus more on proving the short run property than on proving the small separator property or
(2) proving small separator property is somehow possible for $\F$.
However, in the latter case the mentioned VASS has to be constructed by the use of rather different techniques than currently known, as it needs to violate conditions of Theorems~\ref{thm:simple}~and~\ref{thm:advanced}.

We have not considered VASSes occurring in the most recent papers
proving the \ackermann-hardness~\cite{DBLP:conf/focs/Leroux21,DBLP:conf/focs/CzerwinskiO21},
but it seems to us that these techniques are promising
at least with respect to VASSes occurring in~\cite{DBLP:conf/focs/CzerwinskiO21}.

\subparagraph*{Organisation of the paper}
In Section~\ref{sec:prelim} we introduce preliminary notions and recall standard facts. Then in Section~\ref{sec:simple}
we state and prove our first main result, Theorem~\ref{thm:simple}.
In Section~\ref{sec:applications} we provide two applications of Theorem~\ref{thm:simple},
in Section~\ref{ssec:4vass} to the $4$-VASS described in~\cite{DBLP:conf/concur/Czerwinski0LLM20}
and in Section~\ref{ssec:tower} to VASSes occurring in the paper~\cite{DBLP:conf/stoc/CzerwinskiLLLM19} proving the \tower-hardness
of the reachability problem.
Finally, in Section~\ref{sec:advanced} we prove our second main result, Theorem~\ref{thm:advanced}.

\section{Preliminaries}\label{sec:prelim}

\subparagraph*{Basic notions}
We denote by $\N$ the set of nonnegative integers and by $\N_{+}$ the set of positive integers.
For $a, b \in \N$ by $[a, b]$ we denote the set $\{a, a+1, \ldots, b-1, b\}$.
For a set $S$ we write $|S|$ to denote its size, i.e. the number of its elements.
For two sets $A, B$ we define $A + B = \{a+b \mid a \in A, b \in A\}$ and $AB = \{a \cdot b \mid a \in A, b \in B\}$.
In that context we often simplify the notation and write $x$ instead of the singleton set $\{x\}$,
for example $x + \N y$ denotes the set $\{x\} + \N \{y\}$.
The \emph{description size} of an irreducible fraction $\frac{p}{q}$ is $\max(|p|, |q|)$, for $r \in \Q$
its description size is the description size of its irreducible form.

For a $d$-dimensional vector $x = (x_1, \ldots, x_d) \in \R^d$ and index $i \in [1,d]$ we write $x[i]$ for $x_i$.
For $S \subseteq [1,d]$ we write $\proj_S(x)$ to denote the $|S|$-dimensional vector
obtained from $x$ by removing all the coordinates outside $S$.
The \emph{norm} of a vector $x \in \N^d$ is $\norm(x) = \max_{i \in [1,d]} |x[i]|$.
For $i \in [1,d]$ the elementary vector $e_i$
is the unique vector such that $e_i[j] = 0$ for $j \neq i$ and $e_i[i] = 1$.
By $0^d \in \N^d$ we denote the $d$-dimensional vector with all the coordinates being zero.

\subparagraph*{Vector Addition Systems}
A $d$-dimensional Vector Addition System with States (shortly $d$-VASS or just a VASS) consists of finite set of \emph{states} $Q$
and finite set of \emph{transitions} $T \subseteq Q \times \Z^d \times Q$. 
\emph{Configuration} of a $d$-VASS $V = (Q, T)$ is a pair $(q, v) \in Q \times \N^d$, we often write $q(v)$ instead of $(q, v)$.
For a configuration $c = q(v)$ we write $\state(c) = q$.
For a set of vectors $S \subseteq \N^d$ and state $q \in Q$ we write $q(S) = \{q(v) \mid v \in S\}$.
Transition $t = (p, u, q)$ can be fired in configuration $(r, v) \in Q \times \N^d$ if $p = r$ and $u+v \in \N^d$.
Then we write $p(v) \trans{t} q(u+v)$. The triple $(p(u), t, q(u+v))$ is called an \emph{anchored transition}.
A \emph{run} is a sequence of anchored transitions $\rho = (c_1, u_1, c_2), \ldots, (c_n, u_n, c_{n+1})$.
Such a run $\rho$ is then a run \emph{from} configuration $c_1$ \emph{to} configuration $c_{n+1}$
and \emph{traverses} through configurations $c_i$ for $i \in [2,n]$.
We also say that $\rho$ is from $\state(c_1)$ to $\state(c_{n+1})$.
If there is a run from configuration $c$ to configuration $c'$ we also say that $c'$ is \emph{reachable} from $c$
or $c$ \emph{reaches} $c'$ and write $c \reaches c'$. Otherwise we write $c \nreaches c'$.
The configuration $c$ is the \emph{source} of $\rho$ while configuration $c'$ is the \emph{target} of $\rho$,
we write $\src(\rho) = c$ and $\trg(\rho) = c'$.
For a configuration $c \in Q \times \N^d$ we denote $\post_V(c) = \{c' \mid c \reaches c'\}$
the set of all the configurations reachable from $c$ and $\pre_V(c) = \{c' \mid c' \reaches c \}$
the set of all the configurations which reach $c$.
The \emph{reachability problem for VASSes} given a VASS $V$ and two its configurations $s$ and $t$
asks whether $s$ reaches $t$ in $V$.
For a VASS $V$ by its \emph{size} we denote the total number of bits needed to represent its states and transitions.
A VASS is said to be \emph{binary} if numbers in its transitions are encoded in binary.
Effect of a transition $(c, u, c') \in Q \times \Z^d \times Q$ is the vector $u \in \N^d$.
We extend this notion naturally to anchored transitions and runs, effect of the run $\rho = (c_1, u_1, c_2), \ldots, (c_n, u_n, c_{n+1})$
is equal to $u_1 + \ldots + u_n$.
Vector Addition Systems (VASes) are VASSes with just one state or in other words VASSes without states.
It is well known and simple to show that the reachability problems for VASes and for VASSes are polynomially interreducible.
In this work we focus wlog. on the reachability problem for VASSes.

\subparagraph*{Semilinear sets}
For any vectors $b, v_1, \ldots, v_k \in \N^d$ the set $L = b + \N v_1 + \ldots + \N v_k$ is called a \emph{linear} set.
Then vector $b$ is the \emph{base} of $L$ and vectors $v_1, \ldots, v_k$ are \emph{periods} of $L$.
Set of vectors is \emph{semilinear} if it is a finite union of linear sets.
Set of VASS configurations $S \subseteq Q \times \N^d$ is \emph{semilinear} if it is a finite union of sets of the
form $q_i(S_i) \subseteq Q \times \N^d$, where $S_i$ are semilinear as sets of vectors.
The \emph{size} of a representation of a linear set is the sum of norms of its base and periods.
The \emph{size} of a representation of a semilinear set $\bigcup_i L_i$ is the sum of sizes of representations of the sets $L_i$.
The size of a semilinear set is the size of its smallest representation.

For a $d$-VASS $V = (Q, T)$ and two of its configurations $s, t \in Q \times \N^d$ a set $S \subseteq Q \times \N^d$
of configurations is a \emph{separator for $(V, s, t)$} if it fulfils the following conditions:
1) $s \in S$, 2) $t \not\in S$, 3) $S$ is invariant under transitions of $V$, namely for any $c \in S$ such that $c \trans{t} c'$
for some $t \in T$ we also have $c' \in S$. In our work we usually do not exploit by condition 3) by itself,
but use the facts which are implied by all the conditions 1), 2) and 3) together: $\post(s) \subseteq S$ and $\pre(t) \cap S = \emptyset$.

\subparagraph*{Well quasi-order on runs}
We say that an order $(X, \preceq)$ is a well-quasi order (wqo) if in every infinite sequence $x_1, x_2, \ldots$ of elements of $X$
there is a \emph{domination}, i.e. there exist $i < j$ such that $x_i \preceq x_j$. 

Fix a $d$-VASS $V = (Q, T)$.
We define here a very useful order on runs, which turns out to be a wqo (a weaker version was originally introduced in~\cite{DBLP:journals/tcs/Jancar90}).
For two configurations $p(u), q(v) \in Q \times \N^d$ we write $p(u) \preceq q(v)$ if $p = q$ and $u[i] \leq v[i]$ for each $i \in [1,d]$.
For two anchored transitions in $(c_1, t, c_2), (c'_1, t', c'_2)$ we write $(c_1, t, c_2) \preceq (c'_1, t', c'_2)$
if $t = t'$, $c_1 \preceq c'_1$ and $c_2 \preceq c'_2$ (notice that the last condition is actually implied by the previous two).
For two runs $\rho = m_1 \ldots m_k$ and $\rho' = m'_1 \ldots m'_\ell$, where all $m_i$ for $i \in [1,k]$
and $m'_i$ for $i \in [1,\ell]$ are anchored transitions we write
$\rho \unlhd \rho'$ if and there exists a sequence of indices $i_1 < i_2 < \ldots < i_{k-1} < i_k = \ell$ such that
for each $j \in [1,k]$ we have $m_j \preceq m'_{i_j}$. Notice there that setting $i_k = \ell$ implies that $\trg(\rho) \preceq \trg(\rho')$.
All the subsequent considerations can be analogously applied in the case when we demand $\trg(\rho) = \trg(\rho')$
and $i_1 = 1$, but $i_k$ does not necessarily equal $\ell$, which enforces that $\src(\rho) \preceq \src(\rho')$.
The following claim is a folklore, for a proof see Proposition 19 in~\cite{DBLP:conf/stacs/ClementeCLP17}.

\begin{claim}
Order $\unlhd$ is a wqo on runs with the same source.
\end{claim}

The order $\unlhd$ has a nice property that runs bigger than a fixed one are additive in a certain sense.
The following claim is also a folklore, for a proof see Lemma 23 in~\cite{DBLP:journals/corr/ClementeCLP16}
(arXiv version of~\cite{DBLP:conf/stacs/ClementeCLP17}).

\begin{claim}\label{cl:adding-runs}
Let $\rho$, $\rho_1$ and $\rho_2$ be runs of VASS $V$ with $\src(\rho) = \src(\rho_1) = \src(\rho_2)$
such that $\rho \unlhd \rho_1, \rho_2$,
$\trg(\rho_1) = \trg(\rho) + \delta_1$ and $\trg(\rho_2) = \trg(\rho) + \delta_2$ for some $\delta_1, \delta_2 \in \N^d$.
Then there exist a run $\rho'$ such that $\src(\rho) = \src(\rho')$, $\rho \unlhd \rho'$
and $\trg(\rho') = \trg(\rho) + \delta_1 + \delta_2$.
\end{claim}

The following corollary can be easily shown by induction on $n$.

\begin{corollary}\label{corr:pumping}
Let $\rho \unlhd \rho'$ be runs of VASS $V$ such that
$\src(\rho') = \src(\rho)$ and $\trg(\rho') = \trg(\rho) + \delta$ for some $\delta \in \N^d$.
Then for any $n \in \N$ there exists a run $\rho_n$ of $V$ such that $\src(\rho_n) = \src(\rho)$
and $\trg(\rho_n) = \trg(\rho) + n \delta$.
\end{corollary}

The above notions will be useful in the proof of Theorem~\ref{thm:advanced}.

\section{Main tool}\label{sec:simple}

In this section we state and prove our first main result, Theorem~\ref{thm:simple}.
It is simpler both in the formulation and in the proof from the more involved
Theorem~\ref{thm:advanced}, but sufficient to show applications in Section \ref{sec:applications}.

The statement of Theorem~\ref{thm:simple} may seem at first glance a bit artificially involved
thus before stating it we explain the intuition behind this research direction.
Recently there is a big progress in showing lower complexity bounds for the reachability problem
in VASSes. On the intuitive level proving a new lower bound often boils down to finding
a new source of hardness in VASSes or in other words designing a family of VASSes
which is involved in some sense. In particular these involved VASSes should have
a run from some source configuration to some target configuration,
but each such run should be long. For some reason VASS families proposed in a number of recent constructions
share many similarities. The first one is that the runs from the source to the target usually have some
very specific shape and there is exactly one shortest run. Along this run usually some very restrictive
invariants are kept and any deviation from them results in the impossibility of reaching the target configuration.
To our best knowledge the same schema reappears in all the known lower bound
results, in particular in the recent ones~\cite{DBLP:journals/jacm/CzerwinskiLLLM21,DBLP:conf/concur/Czerwinski0LLM20,DBLP:conf/focs/Leroux21,DBLP:conf/focs/CzerwinskiO21,DBLP:journals/corr/abs-2203-04243}.
In many of these constructions the following invariant is kept: product of some two counters
$\vr{x}$ and $\vr{y}$ equals the other counter $\vr{z}$ where the counter $\vr{x}$ is bounded,
while $\vr{y}$ and $\vr{z}$ are unbounded. In other words the ratio $\vr{y} / \vr{z}$ is fixed and bounded
at some point at the run. The reason behind this phenomenon is the use of so called
multiplication triples technique introduced in~\cite{DBLP:conf/stoc/CzerwinskiLLLM19}.
Theorem~\ref{thm:simple} tries to describe this situation more precisely and state that in some way for VASSes
following this popular line of design we might expect to have troubles with avoiding big separators.
Because assumptions of Theorem~\ref{thm:simple} are actually very strong and restrictive
we formulate also Theorem~\ref{thm:advanced} in Section~\ref{sec:advanced} which tries
to take a bit more flexible approach for making the above described phenomena more precise.
Notice that in Theorem~\ref{thm:simple} we describe VASSes with no run from $s$ to $t$, but
there is a close resemblance to the above described setting as there is a run from $s + e_2$ to $t$.

\begin{theorem}\label{thm:simple}
Let $V$ be a $d$-VASS, $s, t \in Q \times \N^d$ be two its configurations, $q \in Q$ be a state
and a line $\alpha = a + \N \Delta$ for vectors $a, \Delta \in \N^d$ be such that
\begin{enumerate}[(1)]
  \item $\Delta = (\Delta[1], \Delta[2], 0^{d-2}) \in \N^d$,
  \item there is no run from $s$ to $t$,
  \item $q(\alpha) \subseteq \post_V(s)$,
  \item $q(\alpha + \N_{+} e_2) \subseteq \pre_V(t)$.
\end{enumerate}
Then each separator for $(V, s, t)$ contains a period $r \cdot \Delta$ for some $r \in \Q$.
\end{theorem}

\begin{proof}
Consider an arbitrary separator $S$ for $(V, s, t)$. Let $S_q = \{v \mid q(v) \in S, v \in \N^d\}$
be its part devoted to the state $q$ and let $S_q = \bigcup_{i \in I} L_i$,
where $L_i$ are linear sets. As the line $\alpha$ contains infinitely many points
and $\alpha \subseteq S_q$ then some of the linear sets $L_i$ have to contain infinitely many points of $\alpha$.
Denote this $L_i$ by $L$, let $L = b + \N p_1 + \ldots + \N p_k$.
We know that for arbitrarily big $n$ we have $a + n\Delta = b + n_1 p_1 + \ldots + n_k p_k$ for some $n_i \in \N$.
Wlog. we can assume that $n_i > 0$, we just do not write the periods $p_i$ with coefficient $n_i = 0$.
Notice first that for each coordinate $j \in [3,d]$ we have $(a + n\Delta)[j] \leq \norm(a)$.
Therefore for $p_i$ such that $p_i[j] > 0$ for some $j \in [3,d]$ we have $n_i \leq \norm(a)$.
Let $P_0$ be the set of periods $p_i$ such that $p_i[j] = 0$ for all $j \in [3,d]$ and
$P_{\neq 0}$ be the set of the other periods $p_i$.
We thus have
\[
(a - b - \sum_{p_i \in P_{\neq 0}} n_i p_i) + n\Delta = \sum_{p_i \in P_0} n_i p_i,
\]
where the sum $v = (a - b - \sum_{p_i \in P_{\neq 0}} n_i p_i)$ has a norm bounded by
$B = \norm(a) + \norm(b) + k \cdot \norm(a) \cdot \max_{p_i \in P_{\neq 0}} \norm(p_i)$,
which is independent of $n$. If we restrict the equation to the first two coordinates of the considered vectors we have
\begin{equation}\label{eq:pzero}
(v[1], v[2]) + n (\Delta[1], \Delta[2]) = \sum_{p_i \in P_0} n_i (p_i[1], p_i[2]).
\end{equation}
We aim at showing that one of $p_i$ is equal to $r \cdot \Delta$ for some $r \in \Q$. Recall that for each $j \in [3,d]$
we have $\Delta[j] = p_i[j] = 0$, so it is enough to show that $(p_i[1], p_i[2]) = r \cdot (\Delta[1], \Delta[2])$.

We first show the following claim.

\begin{claim}\label{cl:ratio}
For each period $p \in P_0$ we have $\Delta[1] \cdot p[2] \leq \Delta[2] \cdot p[1]$.
\end{claim}

\begin{proof}
The intuitive meaning of $\Delta[1] \cdot p[2] \leq \Delta[2] \cdot p[1]$ is that $\frac{\Delta[1]}{\Delta[2]} \leq \frac{p[1]}{p[2]}$,
however we cannot write the fraction $\frac{p[1]}{p[2]}$ as it might happen that $p[2] =  0$.
Assume towards a contradiction that for some $p \in P_0$ we have $\Delta[1] \cdot p[2] > \Delta[2] \cdot p[1]$.
Let $\delta = \Delta[1] \cdot p[2] - \Delta[2] \cdot p[1] > 0$. Recall that $a + n\Delta \in L$ for some $n \in \N$,
therefore also $a + n\Delta + \Delta[1] \cdot p \in L$ as $p \in P_0$ is a period of $L$.
Then however
\begin{align*}
\Delta[1] \cdot p & = \Delta[1]  \cdot (p[1], p[2], 0^{d-2}) = (\Delta[1] \cdot p[1], \Delta[1] \cdot p[2], 0^{d-2}) \\
& = (\Delta[1] \cdot p[1], \Delta[2] \cdot p[1] + \delta, 0^{d-2}) = p[1] \cdot \Delta + \delta \cdot e_2 \in \N \Delta + \N_{+} e_2.
\end{align*}
Therefore $a + n\Delta + \Delta[1] \cdot p \in b + \N \Delta + \N_{+} e_2$ and it means that $q(a + n\Delta + \Delta[1] \cdot p) \in \pre_V(t)$.
However we know that $q(a + n\Delta + \Delta[1] \cdot p) \in q(S_q) \subseteq S$ and therefore separator $S$ nonempty intersects 
the set $\pre_V(t)$. This is a contradiction with the definition of separator.
\end{proof}

In order to show that $p = r \cdot \Delta$ for some $p \in P_0$ it is sufficient to show that $\Delta[1] \cdot p[2] = \Delta[2] \cdot p[1]$.
Assume towards a contradiction that for all $p \in P_0$ we have $\Delta[1] \cdot p[2] \neq \Delta[2] \cdot p[1]$.
By Claim~\ref{cl:ratio} we know that actually for all $p \in P_0$ we have that $\Delta[1] \cdot p[2] < \Delta[2] \cdot p[1]$.
Thus $p[1] > 0$ and we can equivalently write that for all $p \in P_0$ we have that $\frac{\Delta[2]}{\Delta[1]} > \frac{p[2]}{p[1]}$.
Let $F$ be the maximal value of $\frac{p[2]}{p[1]}$ for $p \in P_0$, clearly $\frac{\Delta[2]}{\Delta[1]} > F$,
so
\begin{equation}\label{eq:f}
\Delta[2] > F \cdot \Delta[1].
\end{equation}

By~\eqref{eq:pzero} we know that
\[
\frac{(v+n\Delta)[2]}{(v+n\Delta)[1]} \leq F,
\]
as $v+n\Delta$ is a positive linear combination of vectors $p \in P_0$ and for each $p \in P_0$
we have $\frac{p[2]}{p[1]} \leq F$.
Therefore it holds
\[
v[2] + n\Delta[2] \leq F(v[1] + n\Delta[1])
\]
and equivalently
\[
n(\Delta[2] - F\Delta[1]) \leq Fv[1] - v[2].
\]
By~\eqref{eq:f} we have that $\Delta[2] - F\Delta[1] > 0$ therefore
\[
n \leq \frac{Fv[1] - v[2]}{\Delta[2] - F\Delta[1]}.
\]
This is in contradiction with the assumption that $n$ can be arbitrarily big
and finishes the proof.
\end{proof}

\begin{remark}
In Theorem~\ref{thm:simple} instead of separator for $(V, s, t)$ one can consider a separator for $(V', s, t)$,
where $V'$ is obtained from $V$ by adding a loop in state $q$ with the effect of decreasing the second counter.
Indeed, it is easy to observe that all the points 1-4 in the theorem statement remain true after substitution of $V$ by $V'$.
Such a version of Theorem~\ref{thm:simple} is a bit more convenient for some of the applications.
\end{remark}

\section{Applications}\label{sec:applications}
In this section we show how Theorem~\ref{thm:simple} can be used to obtain lower bounds
on the separator size. In Section~\ref{ssec:4vass} we prove that using a construction from~\cite{DBLP:conf/concur/Czerwinski0LLM20}
one can obtain a $4$-VASS with separators of at least doubly-exponential size. 
In Section~\ref{ssec:tower} we show that there exist VASSes with separators of arbitrary high elementary size
of a special shape. Existence of separators of arbitrary high elementary size is an easy consequence
of \tower-hardness of VASS reachability problem, we provide in Theorem~\ref{thm:tower-sep}
a concrete instance of such a separator. It is not a big contribution, but the aim of proving Theorem~\ref{thm:tower-sep}
is rather to illustrate that our techniques can be quite easily applied to many existing VASS examples.

\subsection{Doubly-exponential separator in a $4$-VASS}\label{ssec:4vass}
The aim of this section is to show the following theorem.

\begin{theorem}\label{thm:4vass-sep}
There exists a family of binary $4$-VASSes $(V_n)_{n \in \N}$ of size $\Oo(n^3)$
such that for some configurations $s_n, t_n$ of $V_n$ with $\norm(s_n), \norm(t_n) \leq 1$
such that $s_n$ does not reach $t_n$ the smallest
separator for $(V_n, s_n, t_n)$ is of doubly-exponential size wrt. $n$.
\end{theorem}

The rest of this section is devoted to the proof of Theorem~\ref{thm:4vass-sep}.
Our construction is based on the construction of a family of binary $4$-VASSes $(U_n)_{n \in \N}$
with shortest run of doubly-exponential length,
which is described in Section 5 in~\cite{DBLP:conf/concur/Czerwinski0LLM20}.
The picture below illustrates VASS $U_n$, it is taken from~\cite{DBLP:conf/concur/Czerwinski0LLM20} and modified a bit.
The proof follows a very natural line and is not a big challenge.
It boils down to performing a small modification to $4$-VASSes from~\cite{DBLP:conf/concur/Czerwinski0LLM20}
in order to ensure conditions of Theorem~\ref{thm:simple} and then checking that indeed these
conditions are satisfied. We sketch here the idea behind the construction of the mentioned family of $4$-VASSes
only into such an extend that we can explain the proof of Theorem~\ref{thm:4vass-sep},
for details we refer to~\cite{DBLP:conf/concur/Czerwinski0LLM20}.

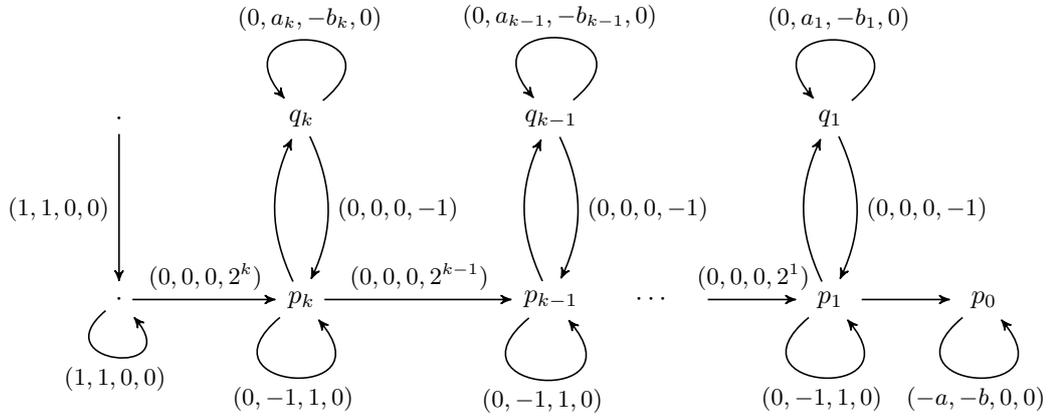
\begin{tikzpicture}[->,>=stealth',shorten >=1pt,auto,node distance=2.4cm,semithick]

\node(sp) {$\cdot$};
\node(s) [above of = sp] {$\cdot$};

\node (pk) [right of = sp] {$p_k$};
\node (qk) [above of = pk] {$q_k$};

\node (pk1) [right = 2.5cm of pk] {$p_{k-1}$};
\node (qk1) [above of=pk1] {$q_{k-1}$};

\node (dotst) [right = 0.5cm of pk1] {$\cdots$};
\node (dots) [right = 0.1cm of dotst] {};

\node (p1) [right = 1.2cm of dots] {$\ p_1\ $};
\node (q1) [above of=p1] {$\ q_1\ $};

\node (e) [right = 1.2cm of p1] {$\ p_0\ $};

\path[->]
(s) edge[->] node[left] {\small $(1,1,0,0)$} (sp)
(sp) edge[->, in=310, out=220, loop] node[below] {\small $(1,1,0,0)$}(sp)
(sp) edge[->] node[above] {\small $(0,0,0, 2^k)$} (pk)
(pk) edge[->, in=310, out=220, loop] node[below] {\small $(0,-1,1,0)$}(pk)
(qk) edge[->, in=130, out=40, loop] node[above] {\small $(0, a_k,-b_k,0)$}(qk)
(pk) edge[->, bend left=25]  (qk)
(qk) edge[->, bend left=25] node[right] {\small $(0,0,0,-1)$} (pk)
(pk) edge[->] node[above] {\small $(0,0,0, 2^{k-1})$} (pk1)
(pk1) edge[->, in=310, out=220, loop] node[below] {\small $(0,-1,1,0)$}(pk1)
(qk1) edge[->, in=130, out=40,loop] node[above] {\small $(0, a_{k-1},-b_{k-1},0)$}(qk1)
(pk1) edge[->, bend left=25] (qk1)
(qk1) edge[->, bend left=25] node[right] {\small $(0,0,0,-1)$} (pk1)
(dots) edge[->] node[above] {\small $(0,0,0, 2^{1})$} (p1)
(p1) edge[->, in=310, out=220,loop] node[below] {\small $(0,-1,1,0)$}(p1)
(q1) edge[->, in=130, out=40, loop] node[above] {\small $(0, a_{1},-b_{1},0)$}(q1)
(p1) edge[->, bend left=25]  (q1)
(q1) edge[->, bend left=25] node[right] {\small $(0,0,0,-1)$} (p1)
(p1) edge[->] node[left] {\small } (e)
(e) edge[->, in=310, out=220, loop] node[below] {\small $(-a,-b,0,0)$}(e)
;
\end{tikzpicture}

The construction relies on the following lemma (Lemma 12 in~\cite{DBLP:conf/concur/Czerwinski0LLM20}).

\begin{lemma}\label{lem:fractions}
For each $n \geq 1$ there are $n$ rational numbers
\[
1 < f_1 < \ldots < f_n = 1 + \frac{1}{4^n}
\]
of description size bounded by $4^{n^2+n}$, such that the description size of $f$
defined as
\begin{equation}\label{eq:fractions}
f = (f_n)^{2^n} \cdot \ldots \cdot (f_2)^{2^2} \cdot (f_1)^{2^1}
\end{equation}
is bounded by $4^{2(n^2+n)}$.
\end{lemma}

The VASSes $U_n$ are constructed as follows.
Valuation of the four counters $(\vr{x}_1, \vr{x}_2, \vr{x}_3, \vr{x}_4)$ in the distinguished \emph{initial} state is initially $0^4$,
the run is \emph{accepting} if it finishes in the distinguished \emph{final} state also with valuation $0^4$.
Each run of $U_n$ consists of $n+2$ phases: the \emph{initial phase},
$n$ phases corresponding to fractions $f_n, f_{n-1}, \ldots, f_2, f_1$, respectively and the \emph{final phase}.
In each run after the initial phase the counter valuation is equal to $(N, N, 0, 0)$ for some nondeterministically guessed value $N \in \N$.
In every accepting run the phase corresponding to the fraction $f_i$ results only in multiplying
the second counter $\vr{x}_2$ by a value $f_i^{2^i}$. Therefore in such an accepting run
counter valuation after phases corresponding to fractions $f_n, \ldots, f_i$ is the following
\[
(\vr{x}_1, \vr{x}_2, \vr{x}_3, \vr{x}_4) = (N, N \cdot f_n^{2^n} \cdot \ldots \cdot f_i^{2^i}, 0, 0).
\]
In particular after all the $n$ phases corresponding to fractions the second counter is equal to
\[
N \cdot f_n^{2^n} \cdot \ldots \cdot f_1^{2^1} = N \cdot f,
\]
where the equality follows from Lemma~\ref{lem:fractions}.
Let $p_i$ be the state of $U_n$ after the initial phase, $n-i$ phases corresponding to fractions $f_n, \ldots, f_{i+1}$
and just before the phase corresponding to fraction $f_i$ (in the case when $i > 0$).
It is important that due to Claim 15 in~\cite{DBLP:conf/concur/Czerwinski0LLM20} any reachable configuration
of the form $p_i(\vr{x}_1, \vr{x}_2, \vr{x}_3, \vr{x}_4)$ satisfies $\vr{x}_2 \leq f_n^{2^n} \cdot \ldots \cdot f_i^{2^i} \cdot \vr{x}_1$.
In particular for any reachable configuration $p_0(\vr{x}_1, \vr{x}_2, \vr{x}_3, \vr{x}_4)$ we have that
$\vr{x}_2 \leq f \cdot \vr{x}_1$, as $f = f_n^{2^n} \cdot \ldots \cdot f_1^{2^1}$.
Let $f = \frac{a}{b}$. In the final phase the transition with the effect $(-b, -a, 0, 0)$ is applied in a loop. One can easily see that in order
to reach counter values $0^4$ after the final phase we need to have in the state $p_0$ values satisfying
the equality $\vr{x}_2 = f \cdot \vr{x}_1$. Similarly any configuration $q_{n-1}(\vr{x}_1, \vr{x}_2, \vr{x}_3, \vr{x}_4)$
on an accepting run needs to satisfy $\vr{x}_2 = f_n^{2^n} \cdot \vr{x}_1$. Let $f_i = \frac{a_i}{b_i}$ for all $i \in [1,n]$.
Then in any accepting run $\vr{x}_1$ needs to be divisible by $b_n^{2^n}$, which is a doubly-exponential number wrt. $n$. This forces
any accepting run of $U_n$ to be doubly-exponential.

We slightly modify $4$-VASSes $U_n$ in order to obtain $4$-VASSes $V_n$, which fulfil conditions of Theorem~\ref{thm:simple}.
VASS $V_n$ is obtained from $U_n$ by adding at the end of $U_n$ two instructions:
1) decrease of $\vr{x}_2$ by $1$; and then 2) a loop, which decreases $\vr{x}_2$ by an arbitrary nonnegative number.
We first show that $V_n$ indeed satisfies conditions of Theorem~\ref{thm:simple} and then we argue how this finishes
the proof of Theorem~\ref{thm:4vass-sep}.
Let $q_\inp$ be the initial state of VASS $U_n$, $q_\out$ be its final state and $q_\last$ be the state after the final phase,
but before applying the above mentioned decrements of $\vr{x}_2$.
We set $s_n = q_\inp(0^4)$ and $t_n = q_\out(0^4)$. Let $N = \prod_{i=1}^n b_i^{2^i}$.
We set $\Delta = (N, N \cdot f_n^{2^n}, 0, 0) \in \N^4$ and fix state $q = q_{n-1}$.
We claim that VASS $V_n$ together with two configurations $s_n$ and $t_n$,
state $q$ and vector $\Delta$ satisfies conditions of Theorem~\ref{thm:simple}.

Condition (1) is satisfied trivially as $\Delta = (N, N \cdot f_n^{2^n}, 0, 0)$.
In order to see that (2) is satisfied recall that if $V_n$ reaches $q_0(\vr{x}_1, \vr{x}_2, \vr{x}_3, \vr{x}_4)$ then
we have that $\vr{x}_2 \leq \frac{a}{b} \vr{x}_1$. In the final phase we subtract in the loop vector $(b, a, 0, 0)$,
but it does not change this inequality. Finally subtracting at least $1$ at $\vr{x}_2$ implies that any
reachable configuration $q_\out(\vr{x}_1, \vr{x}_2, \vr{x}_3, \vr{x}_4)$ satisfies $\vr{x}_2 + 1 \leq \frac{a}{b} \cdot \vr{x}_1$.
This is not true for $\vr{x}_1 = \vr{x}_2 = 0$, which shows that condition (2) is indeed satisfied.
Observe now that for each $k \in \N$ configuration $p_n(kN, kN, 0, 0)$ after the initial phase is reachable from $s_n$.
Thus also configuration $p_{n-1}(kN, kN\cdot f_n^{2^n}, 0, 0)$ is reachable from $s_n$ if we multiply the second counter by $f_n^{2^n}$,
which proves condition (3).
Observe also that starting in a configuration $p_{n-1}(kN, kN\cdot f_n^{2^n}, 0, 0)$ one can reach valuation $q_\last(0^4)$ before
the decrements of $\vr{x}_2$. Therefore after adding $e_2 = (0, 1, 0, 0)$ to both sides
we get that $p_{n-1}(kN, kN\cdot f_n^{2^n}+1, 0, 0) \reaches q_\last(0,1,0,0) \reaches t_n$ for any $k \in \N$.
Because of an option of decreasing $\vr{x}_2$ many times in $q_\last$ we get that 
$p_{n-1}(kN, kN\cdot f_n^{2^n} + \ell, 0, 0) \reaches t_n$ for any $k, \ell \geq 1$, which equivalent to the condition (4).

Therefore by Theorem~\ref{thm:simple} we have that each separator for $(V_n, s_n, t_n)$
contains a period $\Delta \cdot r \in \N^4$ for some $r \in \Q$. Let $\Delta \cdot r = (\vr{x}_1, \vr{x}_2, 0, 0) \in \N^4$.
We know that $\vr{x}_2 = \vr{x}_1 \cdot \Big(\frac{a_n}{b_n}\Big)^{2^n}$, so in order for $\vr{x}_2 \in \N$
we need to have $b_n^{2^n} \mid \vr{x}_1$. This implies that $\vr{x}_1$ is doubly-exponential with respect to $n$.
Therefore size of the period $\Delta \cdot r$ and thus size of the separator is doubly-exponential with respect to $n$.

\subsection{Tower size separators}\label{ssec:tower}
In this section we show the following result.

\begin{theorem}\label{thm:tower-sep}
There exists a family of VASSes $(V_n)_{n \in \N}$ of size polynomial wrt. $n$
such that for some configurations $s_n, t_n$ of $V_n$ with $\norm(s_n) = \norm(t_n) = 0$ the following is true:
\begin{itemize}
  \item $s_n \reaches t_n$, but the shortest run is $n$-fold exponential,
  \item $s_n \nreaches t_n + e$ for some elementary vector $e$ and each separator
  for $(V_n, s_n, t_n + e)$ is of at least $n$-fold exponential size.
\end{itemize}
\end{theorem}

The rest of this section is dedicated to prove Theorem~\ref{thm:tower-sep}.
Construction of VASSes $V_n$ is based on constructions in~\cite{DBLP:conf/stoc/CzerwinskiLLLM19},
but we do not follow exactly~\cite{DBLP:conf/stoc/CzerwinskiLLLM19} in order to avoid
some technicalities and simplify the construction.
As the construction~\cite{DBLP:conf/stoc/CzerwinskiLLLM19} is pretty involved we decided
only to sketch the intuition behind the constructed VASSes and use many of its properties
without a detailed explanation. The main goal of this section is to provide an intuition why Theorem~\ref{thm:simple}
is indeed applicable to that case, for details of the construction we refer to the original paper~\cite{DBLP:conf/stoc/CzerwinskiLLLM19}.
Essentially speaking Theorem~\ref{thm:simple} is applicable to VASSes in which the configurations on the accepting run
are distinguished from the others by keeping some specific ratio of counter values. The bigger the description
size of the ratio the bigger the separator. In VASSes from~\cite{DBLP:conf/stoc/CzerwinskiLLLM19} the description size
is $n$-fold exponential, which implies the lower bound on the size of separators.

We say that a counter is $B$-bounded if its value is upper bounded by $B$ along the whole run.
In~\cite{DBLP:conf/stoc/CzerwinskiLLLM19} counters are assured to be $B$-bounded (for various values of $B$)
in the following way. In order to guarantee that counter $\vr{x}$ initially equal to $0$ is $B$-bounded
we introduce another counter $\vr{\bar{x}}$ which is initialised to $B$ and an invariant
$\vr{a} + \vr{\bar{a}} = B$ is kept throughout the whole run.
The construction of~\cite{DBLP:conf/stoc/CzerwinskiLLLM19} strongly relies on the fact that having
a triple $(\vr{x}, \vr{y}, \vr{z}) = (B, C, BC)$ one can simulate $C / 2$ zero-tests for a $B$-bounded counter.
One zero-test for counter $\vr{a}$ is realised as follows. We decrease $\vr{y}$ by $2$. Then we enter two loops, the first
one with the effect $(-1, 1, -1)$ on counters $(\vr{a}, \vr{\bar{a}}, \vr{z})$
and the second one with the effect $(1, -1, -1)$ on the same counters. It is easy to see that for a $B$-bounded
counter $\vr{a}$ the maximal possible decrease on $\vr{z}$ after these two loops is equal to $2B$ and it can
only be realised if before and after the loops we have $(\vr{a}, \vr{\bar{a}}) = (0, B)$.

Let $3!^n$ be the number $3$ followed by $n$ applications of the factorial function.
For example $3!^0 = 3$, $3!^1 = 6$ and $3!^2 = 720$.
Roughly speaking VASS in~\cite{DBLP:conf/stoc/CzerwinskiLLLM19} consists of a sequence of $n$ gadgets,
such that in every accepting run they compute triples of the form $(3!^i, k_i, 3!^i \cdot k_i)$
for some nondeterministically guessed values $k_i \in \N$.
The first gadget $\B$ computes triple $(3, k_1, 3 \cdot k_1)$ in a very easy way: after increasing its counters by $(3, 0, 0)$
it fires a nondeterministically guessed number $k_1$ of times a loop with the effect equal to $(0, 1, 3)$.
Next we have a sequence of $n-1$ gadgets $\F$, the $i$-th one inputing a triple $(3!^i, k_i, 3!^i \cdot k_i)$
and outputting a triple $(3!^{i+1}, k_{i+1}, 3!^{i+1} \cdot k_{i+1})$ on some other set of three counters.
It is important to mention that correct value of the output triple requires that after the run values of the input triple are all zero.
The last triple $(3!^n, k_n, 3!^i \cdot k_n)$ is used in~\cite{DBLP:conf/stoc/CzerwinskiLLLM19}
to simulate $3!^n$-bounded counters of a counter automaton.
Here however we modify this construction in order to obtain VASSes $V_n$,
which fulfil the conditions of Theorem~\ref{thm:tower-sep}.
As all the triples use different counters the presented VASS has at least dimension $3n$.
In the original construction of~\cite{DBLP:conf/stoc/CzerwinskiLLLM19}
some of the counters were actually reused in order to decrease the dimension.
Here we allow for a wasteful use of counters, this however do not change the idea of the construction.

The family of VASSes $V_n$ is defined as follows.
We distinguish an initial state $q_\inp$ of $V_n$ and a final state $q_\out$ of $V_n$.
For $i \in [1,n]$ let $q_i$ be the state after the $i$-th gadget. From the above description we get that
in state $q_i$ valuation of some three counters is equal to $(3!^i, k_i, 3!^i \cdot k_i)$. Let us denote
these counters $(\vr{x}_i, \vr{y}_i, \vr{z}_i)$. Thus in $q_n$ we reach counter values
$(\vr{x}_n, \vr{y}_n, \vr{z}_n) = (3!^n, k_n, 3!^n \cdot k_n)$ for some $k_n \geq 0$.
Then after state $q_n$ we define a state $q_\decr$ in which we decrease values of $\vr{y}_n$ and $\vr{z}_n$
by applying some nonzero number of zero-tests for $3!^n$-bounded counters.
This operation can be seen as a loop decreasing counters $(\vr{x}_n, \vr{y}_n, \vr{z}_n)$ by $(0, 1, 3!^n)$,
but of course subtracting $(0, 1, 3!^n)$ is not realised by a single transition, but by some smaller sub-gadget of our VASS.
Then we go to a state $q_\out$ in which we can decrease in a loop the counter $\vr{x}_n$
and as well the counter $\vr{y}_n$.
We define $s_n = q_\inp(0^d)$ and $t_n = q_\out(0^d)$ for an appropriate dimension $d$.
We set a vector $e$ to be zero on all the coordinates beside $\vr{y}_n$ on which it is set to be one.
We claim that $V_n$ with configurations $s_n$, $t_n$ and an elementary vector $e$ satisfy
the conditions of Theorem~\ref{thm:tower-sep}. It is easy to see
(assuming all the above remarks about the construction of~\cite{DBLP:conf/stoc/CzerwinskiLLLM19})
that all the accepting runs need to traverse through a configuration $q_n(3!^n, k, 3!^n \cdot k)$ for some $k \geq 1$,
which implies that all the accepting runs have at least $n$-fold exponential length.
It therefore remains to show the second point of Theorem~\ref{thm:tower-sep},
we apply Theorem~\ref{thm:simple} for that purpose.

As counters $\vr{x}_n$, $\vr{y}_n$ and $\vr{z}_n$ are important let us assume wlog. that they correspond to the first three
coordinates in our notation, respectively.
Let $\Delta = (0, 1, 3!^n, 0^{d-3}) \in \N^d$, namely $\Delta[i] = 0$ for all $i \not\in \{\vr{y}_n, \vr{z}_n\}$, $\Delta[\vr{y}_n] = 1$
and $\Delta[\vr{z}_n] = 3!^n$.  Properties of $V_n$ can be summarised in the following claim, which can be
derived from~\cite{DBLP:conf/stoc/CzerwinskiLLLM19}.

\begin{claim}\label{cl:tower}
If $s_n \reaches q_\decr(x, y, z, 0^{d-3})$ then $x = 3!^n$ and $y \leq 3!^n z$.
Moreover for any $y \in \N$ we have $s_n \reaches q_\decr(3!^n, y, 3!^n \cdot y, 0^{d-3})$.
\end{claim}

Now we aim at showing that VASS $V_n$ together with configurations $s_n$, $t_n+e$,
vector $\Delta$ and state $q_\decr$ fulfils conditions of Theorem~\ref{thm:simple}.
It is immediate to see that condition (1) is satisfied as only second and third coordinates in $\Delta$ are nonzero
(we allow for reordering the coordinates in Theorem~\ref{thm:simple} without loss of generality).
In order to show (2) we rely on Claim~\ref{cl:tower}. We have $t_n + e = q_\out(0, 1, 0, 0^{d-3})$,
thus if $s_n \reaches t_n + e$ we need to have $s_n \reaches q_\decr(x, y, 0, 0^{d-3}) \reaches t_n + e$
for some $y > 0$. This is however a contradiction with Claim~\ref{cl:tower}, as then $y \cdot 3!^n > 0$.
By Claim~\ref{cl:tower} we also immediately derive condition (3). To show condition (4) notice
that because of the loop in $q_\decr$ of effect $(0, -1, -3!^n)$ on counters $(\vr{x}_n, \vr{y}_n, \vr{z}_n)$
we have
\[
q_\decr(0, k+\ell, k \cdot 3!^n, 0^{d-3}) \reaches q_\decr(0, \ell, 0^{d-2}) \reaches q_\out(0, \ell-1, 0^{d-2})
\reaches q_\out(0, 1, 0^{d-2}) = t_n + e 
\]
for any $k, \ell \geq 1$. This shows that indeed Theorem~\ref{thm:simple} can be applied to $V_n$.
Thus each separator for $(V_n, s_n, t_n + e)$ contains a period of a form $(r, r \cdot 3!^n, 0, 0^{d-3}) \in \N$.
As $r \in \N$ the number $r \cdot 3!^n$ is $n$-fold exponential and thus the size of any separator for $(V_n, s_n, t_n+e)$
is $n$-fold exponential, which finishes the proof of Theorem~\ref{thm:tower-sep}.

\section{Generalised theorem}\label{sec:advanced}
Before we start introducing the notions needed for Theorem~\ref{thm:advanced} we motivate
the need of a generalisation of Theorem~\ref{thm:simple}. 
Formulation of Theorem~\ref{thm:simple} is rather simple and it is sufficient for our applications, but it may seem a bit arbitrary.
As mentioned in the paragraph before Theorem~\ref{thm:simple} it is well suited to the situation
when the ratio $\frac{x}{y}$ between some two counters $x$ and $y$ is fixed and bounded at some specific moment at the run.
We can however very easily imagine that in some other, but very related VASS 
some a bit more involved ratio is kept fixed. This can be a ratio of a form
$\frac{x_1 + x_2}{y_1 + y_2 + y_3}$ or something similar. It looks very natural that we may
keep constant a ratio not between two counter values, but between sums of a few counter values.
This motivates introduction of the linear functions defined below and use of them in Theorem~\ref{thm:advanced}.
When we deal with more than two counters we cannot easily speak about lines, which was natural in Theorem~\ref{thm:simple}.
This is the reason why we are forced to use a more abstract language in order to be prepared
for the above mentioned simple applications. Theorem~\ref{thm:advanced} can be seen
as a more powerful tool than Theorem~\ref{thm:simple}. At the moment we do not see
any applications in which Theorem~\ref{thm:advanced} is needed. However in our opinion
it is important to show that the techniques presented in Theorem~\ref{thm:simple} can be quite
easily extended to stronger Theorem~\ref{thm:advanced}.

Before stating Theorem~\ref{thm:advanced} we introduce a few notions.

\subparagraph*{Greatest common divisors}
We state here a fact about greatest common divisors, which is helpful in the sequel.
By $\gcd(a_1, \ldots, a_k)$ we denote the greatest common divisor of all the numbers $a_1, \ldots, a_k$.

%

\begin{claim}\label{cl:gcd}
For all natural numbers $a_1, \ldots, a_k \leq M$ and for each $S \geq k(M^2-M)$ which is divisible by $\gcd(a_1, \ldots, a_k)$
there exist nonnegative coefficients $b_1, \ldots, b_k \in \N$ such that $S = a_1 b_1 + \ldots + a_k b_k$.
\end{claim}

\begin{proof}
The following fact is called Bezout's identity and can be easily proved by induction on $k$: 
for each $a_1, \ldots, a_k \in \N$ there exist some coefficients $b_1, \ldots, b_k \in \Z$ such
that $\gcd(a_1, \ldots, a_k) = \sum_{i=1}^k a_i b_i$.  Let us take such a solution which minimises the sum $\sum_{i: b_i < 0} |b_i|$.
We show that in this solution actually all the $b_i$ are nonnegative. Assume otherwise and let $b_i < 0$ for some $i \in [1,k]$.
As $S \geq k(M^2-M)$ then for some $b_j \in [1,k]$ we have $b_j \geq M$. Then substituting $b_j$ by $b_j - a_i$ and
$b_i$ by $b_i + a_j$ we obtain a solution with smaller $\sum_{i: b_i < 0} |b_i|$, contradiction.
\end{proof}

\subparagraph*{Linear functions}
Let a linear function $\lin: \N^d \to \N$ be of a form $\lin(x_1, \ldots, x_d) = \sum_{i=1}^d n_i x_i$, with all the coefficients $n_i \in \N_+$.
We call a linear function \emph{reduced} if $\gcd(n_1, \ldots, n_d) = 1$, let us assume that $\lin$ is reduced.
Notice that each linear function is a reduced linear function multiplied by some natural number.
Let $M = \max_{i \in [1,d]} n_i$.
The \emph{support} of a linear function $\supp(\lin) \subseteq [1,d]$ is the set of coordinates
for which coefficient $n_i$ is nonzero.
Let the set of vectors $\zero(\lin) \subseteq \Z^d$ contain all the vectors $n_j e_i - n_i e_j$
for $i, j \in \supp(\lin)$. Clearly for any $v \in \zero(\lin)$ and any $u \in \N^d$ such that
$u+v \in \N^d$ we have $\lin(u+v) = \lin(u)$.
The following claim tells that set $\zero(\lin)$ in some way spans the set of vectors with the same value
of $\lin$ in case it is big.

\begin{claim}\label{cl:zero-run}
For any $u, v \in \N^d$ such that $\lin(u) = \lin(v) \geq d \cdot M^3$ there is a sequence of
vectors $u = x_0, x_1, \ldots, x_k = v \in \N^d$ such that for all $i \in [1,k]$ we have $x_i - x_{i-1} \in \zero(\lin)$.
\end{claim}

The proof of Claim~\ref{cl:zero-run} can be found in the appendix. It uses Claim~\ref{cl:gcd} and a bit
of other rather simple number theory.

\subparagraph*{Modification of a VASS}

For two linear functions $\lin_1, \lin_2 \in \N^d \to \N$ with disjoint supports
and a VASS $V$ with state $q$ we define VASS $V^q_{\lin_1, \lin_2}$
as $V$ with additional transitions whose aim is to be able to increase the ratio $\frac{\lin_1(\cdot)}{\lin_2(\cdot)}$,
but never decrease it.
Let $\lin_1(x_1, \ldots, x_d) = \sum_{i=1}^d n_{i,1} x_i$ and let $\lin_2(x_1, \ldots, x_d) = \sum_{i=1}^d n_{i,2} x_i$.
The set of states of $V^q_{\lin_1, \lin_2}$ is inherited from $V$ similarly to the set of transitions of $V$.
We additionally add to $V$ transitions of the form $(q, v, q)$, which are loops in the state $q \in Q$, of the following form:
\begin{enumerate}
  \item for each coordinate $i \in \supp(\lin_2)$ add a loop which decreases coordinate $i$ by one
  \item for each coordinate $i \not\in \supp(\lin_1) \, \cup \, \supp(\lin_2)$ add two loops: one, which increases
  and one which decreases coordinate $i$ by one
  \item for each $v \in \zero(\lin_1)$ which is nonzero only at $\supp(\lin_1)$
  and similarly for each $v \in \zero(\lin_2)$ which is nonzero only at $\supp(\lin_2)$ add a loop with effect $v$. 
\end{enumerate}

Notice that condition 2 means that we can freely modify in $V^q_{\lin_1, \lin_2}$ the coordinates outside $\supp(\lin_1) \, \cup \, \supp(\lin_2)$
We are ready to state the theorem.

\begin{theorem}\label{thm:advanced}
Let $\lin_1, \lin_2: \N^d \to \N$ be two reduced linear functions with disjoint supports.
Let $V = (Q, T)$ be a $d$-VASS, $q \in Q$ be its state, $s, t \in Q \times \N^d$ be two its configurations
and $R \in \Q$ be a rational number.
Assume that
\begin{enumerate}[(1)]
  \item for each $v \in \N^d$ if $s \reaches q(v)$ then $\lin_1(v) \geq R \cdot \lin_2(v)$,
  \item for each $u, v \in \N^d$ if $q(v) \reaches t+u$ then $\lin_1(v-u) \leq R \cdot \lin_2(v-u)$,
  \item each run from $\state(s)$ to $\state(t)$ traverses through a configuration $c$ with $\state(c) = q$,
  \item the set $\{\proj_I(v) \mid s \reaches q(v) \reaches t\}$ is infinite, where $I = \supp(\lin_1) \cup \supp(\lin_2)$.
\end{enumerate}
Then for any $i \in \supp(\lin_2)$ there is no run from $s$ to $t+e_i$ in $V' = V^q_{\lin_1, \lin_2}$ and each
separator for $(V', s, t+e_i)$ contains a period $p$ such that $\proj_I(p) \neq 0$ and $\lin_1(p) = R \cdot \lin_2(p)$.
\end{theorem}

\begin{proof}
Due to condition (4) in the theorem statement there exists an infinite sequence of vectors $v_i \in \N^d$ with $\proj_I(v_i)$ pairwise different
such that $s \reaches q(v_i) \reaches t$. Let $\rho^1_i$ be the corresponding runs from $s$ to $q(v_i)$
and $\rho^2_i$ be the corresponding runs from $q(v_i)$ to $t$.
Recall that $\unlhd$ is a well-quasi order and its modified version (denoted here $\unlhd'$) with comparison on sources instead of targets
is also a well-quasi order.
Therefore there exist $i < j$ such that $\rho^1_i \unlhd \rho^1_j$ and $\rho^2_i \unlhd' \rho^2_j$.
Let $\Delta = v_j - v_i \in \N^d$. Clearly $\proj_I(\Delta) \neq 0$, as $\proj_I(v_i) \neq \proj_I(v_j)$. Let $a = v_i$.
By Corollary~\ref{corr:pumping} we get that for any $n \in \N$ there is a run of $V$ from $s$ to $a + n\Delta$.
An analogous reasoning with targets changed to sources shows that for any $n \in \N$ there is a run in $V$
from $a + n\Delta$ to $t$. Therefore for any $n \in \N$ we have $s \reaches q(a + n\Delta) \reaches t$.
By conditions (1) and (2) in the theorem statement we get that $\lin_1(a + n \Delta) = R \cdot \lin_2(a + n \Delta)$
for any $n \in \N$. In consequence $\lin_1(\Delta) = R \cdot \lin_2(\Delta)$. Notice that both $\lin_1(\Delta), \lin_2(\Delta) > 0$,
because $\proj_I(\Delta) \neq 0$.

We first show that $s \nreaches t+e_i$ in $V'$. 
Assume towards a contradiction that $s \reaches t+e_i$. By condition (3) we know
that $s \reaches q(v) \reaches t+e_i$ for some $v \in \N^d$.
Observe first that conditions (1) and (2) still hold for the reachability relation defined in $V'$, as the added loops in state $q$
do not invalidate them (for that purpose we demand in point 3. in the construction of $V'$
that vector $v$ is nonzero only on coordinates in $\supp(\lin_1)$ or only on coordinates in $\supp(\lin_2)$).
Notice that by condition (2) setting $u = e_i$ we have
\[
\lin_1(v) = \lin_1(v - e_i) \leq R \cdot \lin_2(v - e_i) < R \cdot \lin_2(v).
\]
where the first equation follows from the fact that $\lin_1(e_i) = 0$.
On the other hand $s \reaches q(v)$ so by condition (1) we have $\lin_1(v) \geq R \cdot \lin_2(v)$, which is in contradiction
with the above inequality. Thus indeed $s \nreaches t+e_i$.

Now we show that every separator for $(V', s, t+e_i)$ contains an appropriate period $p$.
This proof is quite similar to the proof of Theorem~\ref{thm:simple}, but we need to deal with some more technicalities.
Consider a separator $S = \bigcup_{q \in Q} q(S_q)$ for $(V', s, t+e_i)$.
Clearly $S_q = \bigcup_{j \in J} L_j$, where $L_j$ are linear sets, needs to contain all the vectors $a + n \Delta$ for $n \in \N$.
Let $L$ be one of the finitely many linear sets $L_j$, which contains infinitely many vectors among $\{a + n \Delta \mid n \in \N\}$.
Let $L = b + \N p_1 + \ldots + \N p_k$ and $P$ be the set of periods $\{p_1, \ldots, p_k\}$.
We aim at showing that $L$ contains a period $p$ such that $\proj_I(p) \neq 0$ and $\lin_1(p) = R \cdot \lin_2(p)$,
namely $\lin_1(p) \cdot \lin_2(\Delta) = \lin_2(p) \cdot \lin_1(\Delta)$, as $R = \frac{\lin_1(\Delta)}{\lin_2(\Delta)}$.

We first prove a claim analogous to Claim~\ref{cl:ratio}.
This one is however much more challenging to prove. For its purpose we have defined $V'$ so intricately
with the additional loops.

\begin{claim}\label{cl:ratio2}
For each period $p \in P$ we have $\lin_1(\Delta) \cdot \lin_2(p) \leq \lin_2(\Delta) \cdot \lin_1(p)$.
\end{claim}

Claim~\ref{cl:ratio2} is proven in the appendix, it uses the nontrivial Claim~\ref{cl:zero-run}.

Now we follow the lines of the proof of Claim~\ref{cl:ratio}.
Let $P_\emptyset$ be the set of all periods in $P$ for which $\proj_I(p) = \emptyset$
and $P_\nemp = P \setminus P_\emptyset$.
Recall that we want to show existence of a period $p \in P_\nemp$ fulfilling $\lin_1(p) \cdot \lin_2(\Delta) = \lin_2(p) \cdot \lin_1(\Delta)$.
Assume towards a contradiction that there is no such period.
Therefore by Claim~\ref{cl:ratio2} each period $p \in P_\nemp$ satisfies $\lin_1(\Delta) \cdot \lin_2(p) < \lin_2(\Delta) \cdot \lin_1(p)$.
In particular for each period $p \in P_\nemp$ we have $\lin_1(p) > 0$,
thus we can equivalently write that for all $p \in P_\nemp$ it holds
\[
\frac{\lin_2(p)}{\lin_1(p)} < \frac{\lin_2(\Delta)}{\lin_1(\Delta)}.
\]
Let $F$ be the maximal value of $\frac{\lin_2(p)}{\lin_1(p)}$ for $p \in P_\nemp$,
clearly $\frac{\lin_2(\Delta)}{\lin_1(\Delta)} > F$, so
\begin{equation}\label{eq:f2}
\lin_2(\Delta) > F \cdot \lin_1(\Delta).
\end{equation}
Recall now that for arbitrary big $n$ we have that $a + n\Delta \in b + \N p_1 + \ldots + \N p_k$,
thus $(a - b) + n\Delta \in \N p_1 + \ldots + \N p_k$. Let $v = a - b$.
We have then that
\[
(\lin_1(v+n\Delta), \lin_2(v+n\Delta)) = \sum_{p_i \in P_\nemp} n_i (\lin_1(p), \lin_2(p)).
\]
By the above we know that
\[
\frac{\lin_2(v+n\Delta)}{\lin_1(v+n\Delta)} \leq F,
\]
as $v+n\Delta$ is a positive linear combination of periods from $P$,
for each $p \in P_\emptyset$ we have $\lin_1(p) = \lin_2(p) = 0$ and for each $p \in P_\nemp$
we have $\frac{\lin_2(p)}{\lin_1(p)} \leq F$.
Therefore
\[
\lin_2(v) + n \lin_2(\Delta) \leq F(\lin_1(v) + n \lin_1(\Delta))
\]
and equivalently
\[
n (\lin_2(\Delta) - F \cdot \lin_1(\Delta)) \leq F \cdot \lin_1(v) - \lin_2(v).
\]
By~\eqref{eq:f2} we have that $\lin_2(\Delta) - F \cdot \lin_1(\Delta) > 0$ therefore
\[
n \leq \frac{F \cdot \lin_1(v) - \lin_2(v)}{\lin_2(\Delta) - F \cdot \lin_1(\Delta)}.
\]
This is in contradiction with the fact that $n$ can be arbitrarily big
and finishes the proof.
\end{proof}


\paragraph*{Acknowledgements}
We thank S{\l}awomir Lasota and Michał Pilipczuk for inspiring discussions.

\bibliographystyle{plain}
\bibliography{citat}

\appendix

\section{Missing proof}

\begin{proof}[Proof of Claim~\ref{cl:zero-run}]
We prove the claim by induction on $d$. For $d = 1$ clearly $u = v$ and there is nothing to show.
Let us assume that the claim holds for $d-1$, our aim is to prove it for $d$.
Adding a vector $y \in \zero(\lin)$ to $x_{i-1}$ in order to obtain $x_i = x_{i-1} + y \in \N^d$ we call a \emph{step}.
Clearly it is enough find a sequence of steps from $u$ to $v$.
The plan is to apply first such a sequence of steps from $u$ to some $u'$ such that $u'[i] = v[i]$ for some $i \in [1,d]$
and then show by induction assumption that a sequence of steps from $u'$ to $v$ exists as well.

For a subset of coordinates $I \subseteq [1,d]$ and $x \in \N^d$ by $\lin_I(x)$ we denote $\sum_{i \in I} n_i x[i]$.
As $\lin(v) \geq d \cdot M^3$ there exists some $j \in [1,d]$ such that $\lin_{[1,d] \setminus \{j\}}(v) \geq (d-1) M^3$.
Assume wlog. that $j = d$, so
\begin{equation}\label{eq:linv}
\lin_{[1,d-1]}(v) \geq (d-1) M^3.
\end{equation}
We first aim to reach $u''$ such that $u''[d] - v[d] \geq M^2$. Clearly until for some $i \neq d$
we have $u[i] \geq M$ we can apply the step $n_i e_d - n_d e_i$ to $u$ and increase value of $u[d]$. We continue this
until we reach some $u''$ with $\sum_{i=1}^{d-1} u''[i] < (d-1) M$.
This is indeed possible as $\sum_{i=1}^{d-1} u''[i] \geq (d-1) M$ implies that for some $i \neq d$ we have $u''[i] \geq M$. 
Then we have that $\lin_{[1,d-1]}(u'') < (d-1) M^2$ as $M = \max_{i \in [1,d]} n_i$, so $n_d \cdot u''[d] > dM^3 - (d-1)M^2$.
By~\eqref{eq:linv} we have $n_d \cdot v[d] \leq dM^3 - (d-1) M^3 = M^3$.
Therefore $n_d (u''[d] - v[d]) > (d-1) (M^3 - M^2)$ and thus $u''[d] - v[d] \geq (d-1) (M^2 - M)$.
By Claim~\ref{cl:gcd} we have that $u''[d] - v[d] = \sum_{i=1}^{d-1} n_i b_i$ for some $b_i \in \N$
(notice that here we use the fact that $\lin$ is reduced and therefore $\gcd(n_1, \ldots, n_{d-1}) = 1$
and divides $u''[d] - v[d]$). Therefore in order to obtain $u'$ such that $u'[d] = v[d]$ we apply for each $i \in [1,d-1]$
exactly $b_i$ number of times the step $n_d e_i - n_i e_d$ to $u''$.
Notice that all the other coordinates beside the $d$-th one increase, so these steps lead
to vectors with nonnegative coordinates. As $\lin_{i \in [1,d-1]}(u') = \lin_{i \in [1,d-1]}(v) \geq (d-1) M^3$ we apply
the induction assumption to show that indeed starting from $u'$ one can reach $v$ by a sequence of steps. This finishes the proof.
\end{proof}

\begin{proof}[Proof of Claim~\ref{cl:ratio2}]
Similarly as in Claim~\ref{cl:ratio} we intuitively mean to show that for all $p \in P$ we have
$\frac{\lin_1(\Delta)}{\lin_2(\Delta)} \leq \frac{\lin_1(p)}{\lin_2(p)}$, but this is not a formally correct
statement as it may be that $\lin_2(p) = 0$.
Assume towards a contradiction that there is a period $p$ such
that $\lin_1(\Delta) \cdot \lin_2(p) > \lin_2(\Delta) \cdot \lin_1(p)$.
Clearly $a \in L$, therefore also $a + mp \in L$ for any $m \in \N$.
We aim at showing that $a + mp \reaches a + n\Delta + e_i$ in $V'$ for some $m, n \in \N$.
This would lead to a contradiction as we know that $a + n\Delta \reaches t$, so also
$a + n\Delta + e_i \reaches t+e_i$. Therefore we would have that $a + mp \in L$ and
also $a + mp \reaches a + n\Delta + e_i \reaches t + e_i$, which is a contradiction
with the definition of the separator.

Let $\lin_2(x_1, \ldots, x_d) = \sum_{i=1}^d n_{i,2} x_i$ and let $M$ be the maximal coefficient in $\lin_2$,
namely $M = \max_{i \in [1,d]} n_{i,2}$. We set $m = dM^3 \cdot \lin_1(\Delta)$ and $n = dM^3 \cdot \lin_1(p)$.
Observe now that $\lin_1(a + mp) = \lin_1(a + n\Delta + e_i)$.
Indeed
\begin{align*}
\lin_1(a + mp) & = \lin_1(v) + m \cdot \lin_1(p) = \lin_1(v) + dM^3 \cdot \lin_1(\Delta) \cdot \lin_1(p) \\
& = \lin_1(v) + n \cdot \lin_1(\Delta) = \lin_1(a + n\Delta + e_i),
\end{align*}
as $i \not\in \supp(\lin_1)$, so $\lin_1(e_i) = 0$.
Let $p = p_1 + p_2 + p_{\trash}$ and $\Delta = \Delta_1 + \Delta_2 + \Delta_{\trash}$,
where $p_1, \Delta_1$ are positive only on $\supp(\lin_1)$,
$p_2, \Delta_2$ are positive only on $\supp(\lin_2)$
and $p_\trash, \Delta_\trash$ are positive only outside $\supp(\lin_1) \, \cup \, \supp(\lin_2)$.
First notice that thanks to transitions in $V'$, which can freely modify coordinates outside $\supp(\lin_1) \, \cup \, \supp(\lin_2)$
we can assume wlog. that $p_\trash = \Delta_\trash = 0$.
Therefore by Claim~\ref{cl:zero-run} and because $\lin_1$ is reduced there is a run in $V'$ from $a + mp = a + mp_1 + mp_2$
to $a + n \Delta_1 + mp_2$, as $\lin_1(a + mp_1) = \lin_1(a + n \Delta_1) \geq dM^3 \geq |\supp(\lin_1)| \cdot M^3$.
We claim now that
\begin{equation}\label{eq:needed}
\lin_2(mp_2) = \lin_2(n \Delta_2 + e_i + v_\trash) \geq dM^3 \geq |\supp(\lin_2)| \cdot M^3
\end{equation}
for some $v_\trash$ positive only on $\supp(\lin_2)$.
This would finalise the argument as then
\[
a + n \Delta_1 + mp_2 \reaches a + n \Delta_1 + n \Delta_2 + e_i + v_\trash \reaches a + n \Delta_1 + n \Delta_2 + e_i = a + n \Delta + e_i,
\]
where the first relation $\reaches$ follows from Claim~\ref{cl:zero-run} and fact that $\lin_2$ is reduced
and the second one follows from existence of transitions in $V'$ which arbitrarily decrease any coordinate in $\supp(\lin_2)$.

Recall first that $\lin_1(\Delta) \cdot \lin_2(p) > \lin_2(\Delta) \cdot \lin_1(p)$, so 
\[
\lin_1(\Delta) \cdot \lin_2(p) - \lin_2(\Delta) \cdot \lin_1(p) \geq 1
\]
and that we set $m = dM^3 \cdot \lin_1(\Delta)$ and $n = dM^3 \cdot \lin_1(p)$.
Therefore
\[
\lin_2(m p_2) - \lin_2(n \Delta_2) = dM^3 \cdot \lin_1(\Delta) \cdot \lin_2(p) - dM^3 \cdot \lin_2(\Delta) \cdot \lin_1(p) \geq dM^3.
\]
So $\lin_2(m p_2) - \lin_2(n \Delta_2 + e_i) \geq dM^3 - M \geq d(M^2 - M)$. Thus by Claim~\ref{cl:gcd}
there exist some $b_1, \ldots, b_d \in \N$ such that $\sum_{j=1}^d n_{j,2} b_j = \lin_2(m p_2) - \lin_2(n \Delta_2 + e_i)$.
Hence $\lin_2(m p_2) = \lin_2(n \Delta_2 + e_i + \sum_{j=1}^n b_j e_j) \geq dM^3$
and we can set $v_\trash =  \sum_{j=1}^n b_j e_j$ thus satisfying~\eqref{eq:needed}
and finishing the proof of the claim.
\end{proof}

\end{document}